\documentclass[letterpaper,10 pt,conference]{ieeeconf}
\IEEEoverridecommandlockouts
\usepackage{amsmath}
\usepackage{amssymb}
\usepackage{mathptmx}
\usepackage{graphicx}
\usepackage[colorlinks=true,urlcolor=blue,citecolor=blue]{hyperref}

\newtheorem{definition}{Definition}
\newtheorem{theorem}{Theorem}
\newtheorem{remark}{Remark}

\title{\LARGE \bf Cooperative Visual Convex Area Coverage using a Tessellation-free Strategy*}

\author{
Sotiris Papatheodorou$^{1}$ and Anthony Tzes$^{1}$%
\thanks{*This work has received funding from the European Union's Horizon 2020 Research and Innovation Programme under the Grant Agreement No.644128, AEROWORKS.}%
\thanks{$^{1}$ The authors are with New York University Abu Dhabi, United Arab Emirates; This work was partially completed while the authors were with University of Patras, Greece. Corresponding author's email: \small\tt{anthony.tzes@nyu.edu}}%
}

\begin{document}
\maketitle
\thispagestyle{empty}
\pagestyle{empty}

\begin{abstract}
The objective in this article is to develop a control strategy for coverage purposes of a convex region by a fleet of Mobile Aerial Agents (MAAs). Each MAA is equipped with a downward facing camera that senses a convex portion of the area while its altitude flight is constrained. Rather than relying on typical Voronoi-like tessellations of the area to be covered, a scheme focusing on the assignment to each MAA of certain parts of the mosaic of the current covered area is proposed. A gradient ascent algorithm is then employed to increase in a monotonic manner the covered area by the MAA-fleet. Simulation studies are offered to illustrate the effectiveness of the proposed scheme.
\end{abstract}


\section{Introduction}
Area coverage over a planar region by ground agents has been studied extensively when the sensing patterns of the agents are circular \cite{Cortes_ESAIMCOCV05,Pimenta_CDC08}. The majority of these techniques is based on a Voronoi or similar partitioning~\cite{Cortes_TRA04,Stergiopoulos_IETCTA10,Arslan_ICRA2016} of the region of interest and the employed control techniques include distributed optimization, model predictive control~\cite{Mohseni_IEEESJ2014} and game theory~\cite{Ramaswamy_ACC2016} among others. There has also been significant work concerning arbitrary sensing patterns~\cite{Stergiopoulos_Automatica13,Kantaros_Automatica15,Panagou_IEEETACCNS2016} while avoiding the usage of Voronoi partitioning~\cite{Stergiopoulos_ICRA14,Bakolas_SCL2016}. Convex as well as non-convex regions of interest have been studied~\cite{Stergiopoulos_IEEETAC15,Alitappeh_SC2016}.

A multitude of algorithms has been developed for mapping by  MAAs~\cite{Renzaglia_IJRR12,Breitenmoser_IROS10,Thanou_ISCCSP14,Torres_ESA2016} relying mostly in Voronoi--based tessellations of the region or path planning. Extensive research has been conducted in the field of area monitoring by MAAs equipped with cameras~\cite{Schwager_ICRA2009,Schwager_IEEE2011}. In these pioneering research efforts, no constraints are imposed on the altitude of the MAAs and overlapping between their covered regions is considered advantageous compared to the same region sensed by a single MAA. There is also work on ground target detection and tracking~\cite{Hu2012_ICARCV} by a team of MAAs as well as the connectivity and energy consumption of MAA networks~\cite{Yanmaz_ICC2012,Messous_WCNC2016}.

This article examines the persistent coverage problem of a convex planar region by a network of MAAs. The MAAs are assumed to have downwards facing visual sensors with arbitrary convex sensing footprints. Thus, both the covered area and the coverage quality of that area are dependent on the altitude of each MAA. MAAs at higher altitudes cover more area but with a lower coverage quality compared to MAAs at lower altitudes. A partitioning scheme of the sensed region, similar to~\cite{Stergiopoulos_ICRA14}, is employed in order to derive a gradient--based distributed control law. This control law leads the network to a locally optimal configuration with respect to a combined coverage--quality criterion. This work is an extension of \cite{Papatheodorou2017_RAS} from circular sensing footprints to any convex sensing footprint. Algorithmic implementations of the proposed control scheme are provided as free and open--source (Matlab--based) code in an online repository.

The problem statement and the optimization criterion are formulated in Section \ref{sec:problem_statement}. The coverage quality function and its properties are presented in Section \ref{sec:quality} while the sensed space partitioning scheme is defined in Section \ref{sec:partitioning}. The distributed control law is derived in Section \ref{sec:control}. Simulation studies highlighting the efficiency of the proposed control law in comparison to other control schemes are provided in Section \ref{sec:simulations} followed by concluding remarks.

\section{Problem Statement}
\label{sec:problem_statement}
We assume a compact and convex region $\Omega \subset \mathbb{R}^2$ which is under surveillance. We also assume a team of $n$ MAAs, each positioned at $X_i = \left[ x_i~y_i~z_i \right]^T, ~i \in I_n$ and having an orientation $\theta_i, ~i \in I_n$ with respect to the yaw axis, where $I_n = \left\{ 1, \dots ,n \right\}$. In addition, we define the vector $q_i = [ x_i ~ y_i ]^T, ~q_i \in \Omega$, denoting the projection of each MAA on the ground. Each MAA can fly between two predefined minimum and maximum altitudes $z_i^{\min}$ and $z_i^{\max}$ respectively, with $z_i^{\min} < z_i^{\max}$, thus $z_i \in [z_i^{\min}, ~z_i^{\max}], ~i \in I_n$. It assumed that $z_i^{\min} > 0, ~\forall i \in I_n$, since setting the minimum altitude to zero could potentially lead to some MAAs crashing on the ground. The minimum altitude $z_i^{\min}$ ensures that the MAAs will fly above ground obstacles, whereas the maximum altitude $z_i^{\max}$ that they will remain within range of their base station.

Instead of using a complete dynamic model, such as the quadrotor helicopter dynamics in \cite{Alexis2011_CEP}, a simplified single integrator dynamic model is used. The MAAs are approximated by point masses, thus each one's kinodynamic model is
\begin{align}
	\nonumber
	\dot{q}_i &= u_{i,q}, ~~q_i \in \Omega, &~u_{i,q} \in \mathbb{R}^2,\\
	\nonumber
	\dot{z}_i &= u_{i,z}, ~~z_i \in [z_i^{\min} ~,~ z_i^{\max}], &~u_{i,z}\in \mathbb{R},\\
	\dot{\theta}_i &= \omega_i, ~~\theta_i \in \mathbb{R}, &~\omega_i \in \mathbb{R}.
	\label{eq:kinematics}
\end{align}
where $\left[u_{i,q}^T, ~u_{i,z}, ~\omega_i\right]^T$ is the corresponding control input for each MAA~(node). In the sequel, all MAAs are assumed to have common minimum $z^{\min}$ and maximum $z^{\max}$ altitudes, although the results produced could be generalized to account for different altitude constraints among MAAs.

Regarding the sensing performance of the MAAs, they are assumed to be equipped with identical, downwards facing visual sensors. Each sensor is able to cover a convex region on the ground. As the altitude of an MAA increases, the area of the region it surveys also increases, as shown in Figure \ref{fig:problem_statement}. Without loss of generality, we define a base sensing pattern $C^b$ as the region surveyed by an MAA positioned at $X_i = \left[ 0~0~z^{\min} \right]^T$ with an orientation $\theta_i = 0$. We assume that there exists a parametrization $\gamma(k)$ of the boundary $\partial C^b$ of this base sensing pattern
\begin{align}
	\partial C^b = \left\{ q \in \mathbb{R}^2 \colon q = \gamma(k), ~k \in K \right\}
	\label{eq:base_parametrization}
\end{align}
where the interval $K$ is dependent on the parametrization and
\begin{equation*}
	\gamma(k) = \left[\begin{array}{c} \gamma_x(k) \\ \gamma_y(k) \end{array}\right], ~k \in K
\end{equation*}
Thus the sensing pattern of an MAA located at $X_i = \left[ x_i~y_i~z_i \right]^T$ with an orientation $\theta_i$ can be derived by scaling according to its altitude, rotating around the origin, and finally translating the base sensing pattern as follows
\begin{align}
	C_i(X_i,\theta_i) = q_i + \textbf{R}(\theta_i) \frac{z_i }{z^{\min}} C^b
	\label{eq:sensing}
\end{align}
where $\textbf{R}(\cdot)$ is the $2\times2$ rotation matrix.
Additionally, we define the minimum base sensing pattern radius as
\begin{equation*}
	r_b^{\min} = \min \left\{ \parallel q \parallel, \forall q \in \partial C_b \right\}.
\end{equation*}
This radius is constant since the base sensing pattern is independent of the agents' states.

As the altitude of an MAA increases, the visual coverage quality of its sensed region decreases. We define a coverage quality function $f(z_i)\colon [z^{\min}, ~z^{\max}] \rightarrow [0, 1]$ which is dependent on the MAA's altitude constraints $z^{\min}$ and $z^{\max}$. This coverage quality is assumed to be uniform throughout the sensed region $C_i$. The greater the value of $f(z_i)$, the better the coverage quality of $C_i$. The exact definition and properties of this coverage quality function are presented in Section \ref{sec:quality}.

Each point $q \in \Omega$ is assigned an importance weight through a space density function $\phi \colon \Omega \rightarrow \mathbb{R}^+$, which expresses the a priori information regarding the region of interest.

Finally, we define the joint coverage-quality objective function
\begin{equation}
	\mathcal{H} \stackrel{\triangle}{=} \int_{\Omega} \max_{i \in I_n} f(z_i) ~\phi(q) ~dq.
	\label{eq:objective_initial}
\end{equation}

\begin{figure}[htb]
	\centering
	\includegraphics[width=0.4\textwidth]{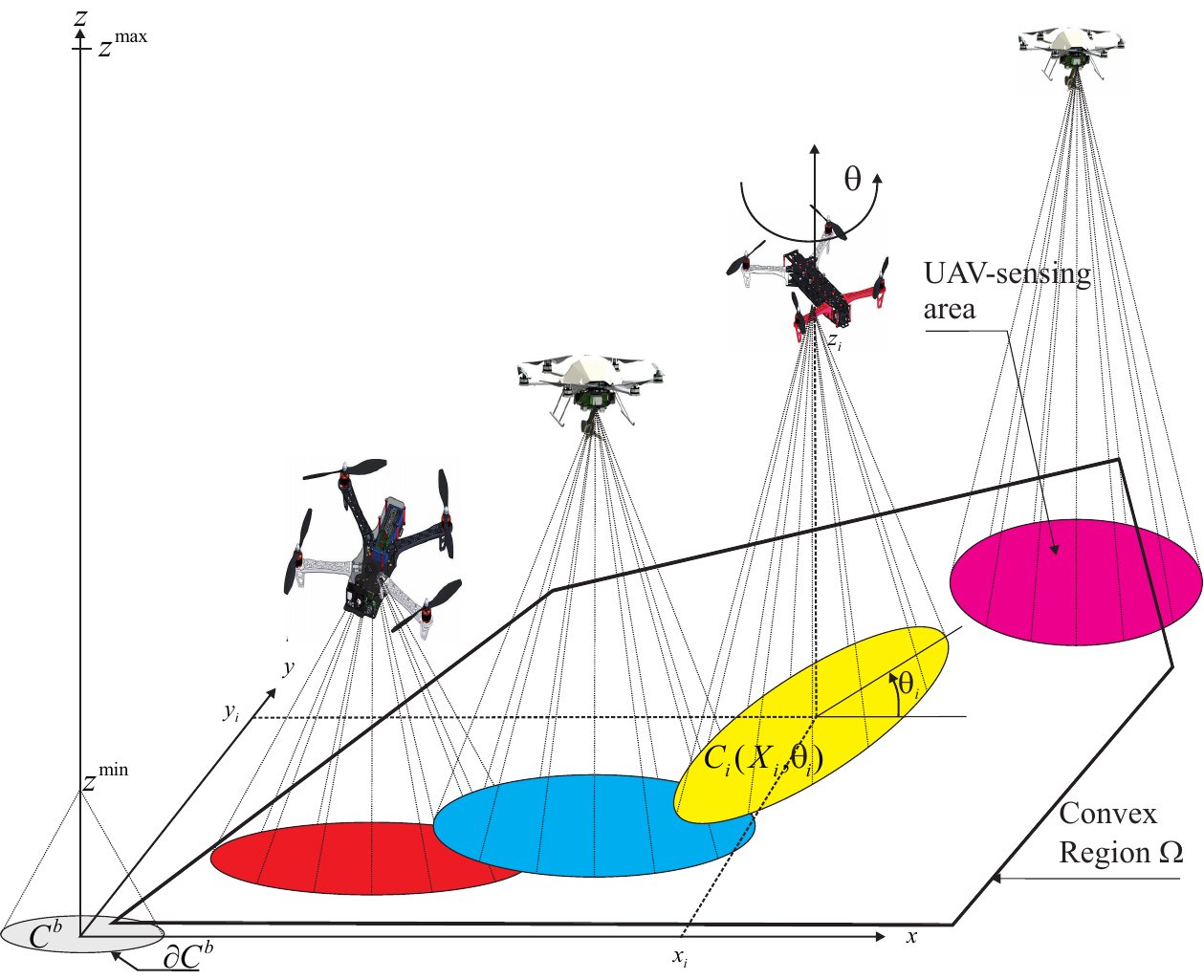}
	\caption{MAA-visual area coverage concept}
	\label{fig:problem_statement}
\end{figure}

\section{Coverage Quality Function}
\label{sec:quality}
A uniform coverage quality throughout each sensed region $C_i$ can be used to effectively model downward facing cameras~\cite{DiFranco_JIRS2016,Avellar_S2015} that provide uniform quality throughout the whole image. The uniform coverage quality function $f(z_i)\colon [z^{\min}, ~z^{\max}] \rightarrow [0, 1]$ was chosen to be
\begin{equation}
	f(z_i) = \left \{
	\begin{aligned}
		&~\frac{\left( \left( z_i - z^{\min} \right)^2 - \left( z^{\max} - z^{\min} \right)^2 \right)^2}{\left( z^{\max} - z^{\min} \right)^4}, & ~q \in C_i\\
		&~0, & ~q \notin C_i
	\end{aligned}
	\right.\\
	\label{eq:quality}
\end{equation}

A plot of this function can be seen in Figure \ref{fig:quality_plot} [Left]. This function was chosen so that $f(z^{\min}) = 1$, $f(z^{\max}) = 0$ and because $f(z_i)$ is first order differentiable with respect to $z_i$, or $\frac{\partial f(z_i)}{\partial z_i}$ exists within $C_i$, which is a property that will be required when deriving the control law in Section \ref{sec:control}.

\begin{figure}[htb]
	\centering
	\includegraphics[width=0.23\textwidth]{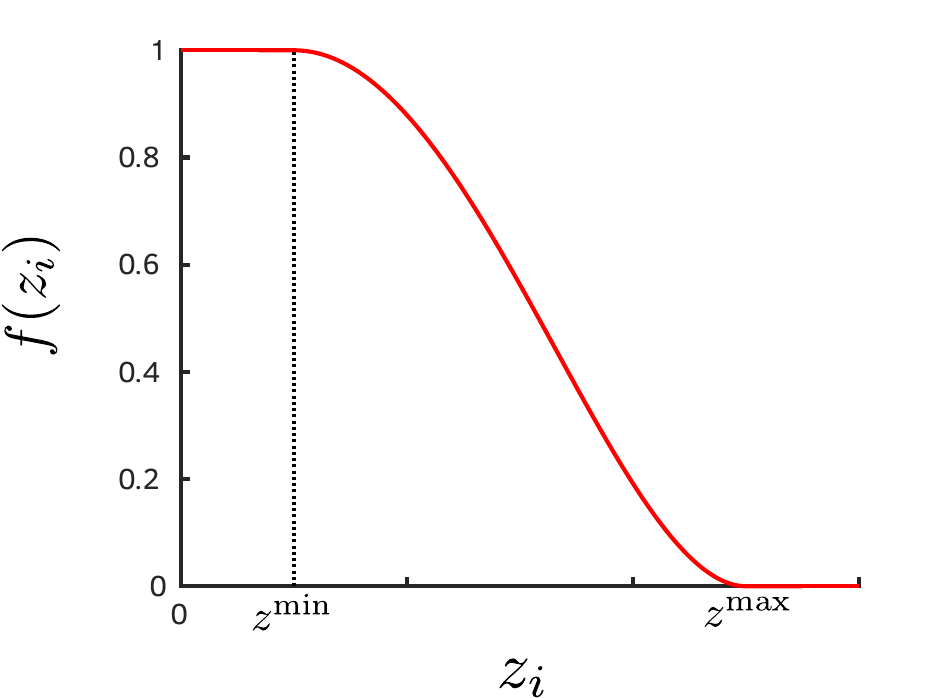}
	\includegraphics[width=0.23\textwidth]{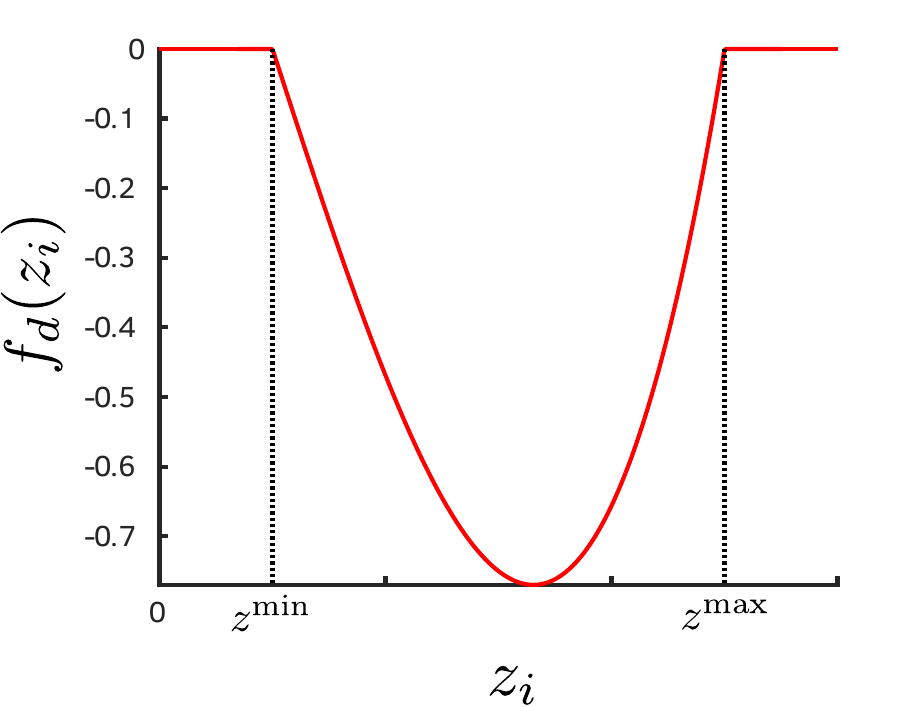}
	\caption{Uniform coverage quality function [Left] and its derivative [Right].}
	\label{fig:quality_plot}
\end{figure}

The derivative $f_d(z_i) \stackrel{\triangle}{=} \frac{\partial f(z_i)}{\partial z_i}\colon [z^{\min}, ~z^{\max}] \rightarrow [f_d^{\min}, 0]$ is evaluated as
\begin{footnotesize}
\begin{equation}
	f_d(z_i) = \left \{
	\begin{aligned}
		&~\frac{4 \left( z_i - z^{\min} \right) \left[ \left( z_i - z^{\min} \right)^2 - \left( z^{\max} - z^{\min} \right)^2 \right]}{\left( z^{\max} - z^{\min} \right)^4}, & ~q \in C_i\\
		&~0, & ~q \notin C_i
	\end{aligned}
	\right.\\
	\label{eq:quality_derivative}
\end{equation}
\end{footnotesize}
where $f_d^{\min} = f_d\left( z^{\min} + \frac{\sqrt{3}}{3} \left(z^{\max}-z^{min}\right) \right) = -\frac{8\sqrt{3}}{9\left(z^{\max}-z^{min}\right)}$. A plot of this function can be seen in Figure \ref{fig:quality_plot} [Right].

\section{Sensed Space Partitioning}
\label{sec:partitioning}
The assignment of responsibility regions to the nodes is achieved in a manner similar to \cite{Stergiopoulos_ICRA14}, where only the subset of $\Omega$ sensed by the nodes is partitioned. Each node is assigned a cell
\begin{equation}
	W_i \stackrel{\triangle}{=} \left\{q \in \Omega \colon f(z_i) > f(z_j), ~j \neq i \right\}, i \in I_n.
	\label{eq:partitioning}
\end{equation}
These cells do not comprise a complete tessellation of the sensed region because they do not include regions that are sensed by multiple nodes with the same coverage quality. We define the set
\begin{equation*}
	I^l = \left\{ i,j \in I_n, ~i \neq j \colon~ C_i \cap C_j \neq \emptyset ~\wedge~ f(z_i) = f(z_j) = f^l \right\}
\end{equation*}
which contains all the nodes with overlapping sensed regions and common coverage quality $f^l$. In general, there can be multiple such regions, each with a different coverage quality. We denote the number of such regions $L$. These regions are not assigned to any node and are defined as
\begin{equation}
	W_c^l \stackrel{\triangle}{=} \left\{\exists~ i, j \in I^l, ~i \neq j \colon~ q
	\in C_i \cap C_j \right\}, ~l \in I_L.
\end{equation}

Because the coverage quality is uniform, $\partial W_j \cap \partial W_i$ is
either an arc of $\partial C_i$ if $f(z_i) > f(z_j)$ or an arc of $\partial
C_j$ if $f(z_i) < f(z_j)$. If the sensing region of a node $i$ is contained
within the sensing region of another node $j$, i.e. $C_i \cap C_j = C_i$, then
$W_i = C_i$ and $W_j = C_j \setminus C_i$. An example partitioning with all of
the aforementioned cases illustrated can be seen in Figure
\ref{fig:partitioning} [Left], where the boundaries of the sensing regions
$\partial C_i$ are in dashed red and the boundaries of the cells $\partial W_i$
in solid black. Nodes $1$ and $2$ are at the same altitude so their common
covered region $W_c^1$ shown in gray is left unassigned. The sensing region of
node $3$ contains the sensing region of node $4$ and nodes $5, 6$ and $7$
illustrate the general case.

By utilizing this partitioning scheme, the network's coverage performance can be written as
\begin{equation}
	\mathcal{H} = \sum_{i \in I_n} \int_{W_i} f(z_i) ~\phi(q) ~dq ~+ \sum_{l = 1}^L
	\int_{W_c^l} f^l ~\phi(q) ~dq.
	\label{eq:criterion}
\end{equation}

\begin{definition}
	We define the neighbors $N_i$ of node $i$ as
	\begin{equation}
		\label{eq:neighbors}
		N_i \stackrel{\triangle}{=} \left\{ j \neq i \colon C_j \cap C_i \neq \emptyset \right\}.
	\end{equation}
	The neighbors of node $i$ are those nodes that sense at least a part of the region that node $i$ senses. It is clear that, due to the partitioning scheme used, only the nodes in $N_i$ need to be considered when creating $W_i$.
\end{definition}

\begin{remark}
	The aforementioned partitioning is a complete tessellation of the sensed region $\bigcup_{i \in I_n} C_i$. However it is not a complete tessellation of $\Omega$. The neutral region not assigned by the partitioning scheme is denoted as $\mathcal{O} = \Omega \setminus \left(\underset{i \in I_n}{\bigcup} W_i ~\cup~ \overset{L}{\underset{l = 1}{\bigcup}} W_c^l\right)$.
\end{remark}

\begin{figure}[htb]
	\centering
	\includegraphics[width=0.25\textwidth]{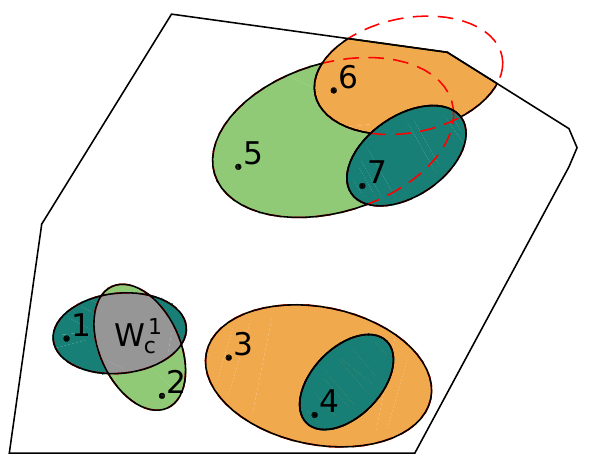}
	\caption{Sensed space partitioning examples.}
	\label{fig:partitioning}
\end{figure}

\section{Spatially Distributed Coordination Algorithm}
\label{sec:control}
Based on the nodes kinodynamics (\ref{eq:kinematics}), their sensing performance (\ref{eq:sensing}) and the coverage criterion (\ref{eq:criterion}), a gradient based control law is designed. The control law utilizes the partitioning scheme (\ref{eq:partitioning}) and results in monotonic increase of the joint coverage--quality criterion $\mathcal{H}$.

\begin{theorem}
In an MAA visual network consisting of nodes governed by the kinematics in (\ref{eq:kinematics}), sensing performance as in (\ref{eq:sensing}) and using the space partitioning (\ref{eq:partitioning}), the control law
\begin{small}
\begin{align}
	\nonumber
	\label{eq:control_xy}
	u_{i,q}
	&= ~\alpha_{i,q} \left[ ~\int\limits_{\partial W_i \cap \partial \mathcal{O}} \upsilon_i^i ~n_i ~f(z_i) ~\phi(q) ~dq ~+\right.\\
	&\qquad \left.\sum\limits_{j \in N_i} ~\int\limits_{\partial W_i \cap \partial W_j} \upsilon_i^i ~n_i ~\left(f(z_i) - f(z_j)\right) ~\phi(q) ~dq \right]
\end{align}
\begin{align}
	\nonumber
	\label{eq:control_z}
	u_{i,z}
	&= ~\alpha_{i,z} \left[ ~\int\limits_{W_i} f^d(z_i) ~\phi(q) ~dq ~+ \int\limits_{\partial W_i \cap \partial \mathcal{O}} \nu_i^i \cdot n_i ~f(z_i) ~\phi(q) ~dq ~+ \right.\\
	&\qquad \left. \sum\limits_{j \in N_i} ~\int\limits_{\partial W_i \cap \partial W_j} \nu_i^i \cdot n_i ~\left(f(z_i)-f(z_j)\right) ~\phi(q) ~dq \right]
\end{align}
\begin{align}
	\nonumber
	\label{eq:control_theta}
	\omega_i
	&= ~\alpha_{i,\theta} \left[ ~\int\limits_{\partial W_i \cap \partial \mathcal{O}} \tau_i^i \cdot n_i ~f(z_i) ~\phi(q) ~dq ~+ \right.\\
	&\qquad \left.\sum\limits_{j \in N_i} ~\int\limits_{\partial W_i \cap \partial W_j} \tau_i^i \cdot n_i ~\left(f(z_i)-f(z_j)\right) ~\phi(q) ~dq \right]
\end{align}
\end{small}
where $\alpha_{i,q}, ~\alpha_{i,z}, ~\alpha_{i,\theta}$ are positive constants, $\upsilon_i^i$, $\nu_i^i$ and $\tau_i^i$ are the Jacobian matrices of the points $q \in \partial W_i$ with respect to $q_i$, $z_i$ and $\theta_i$ as defined in equations (\ref{eq:jacobian_xy}), (\ref{eq:jacobian_z}) and (\ref{eq:jacobian_theta}) respectively and $n_i$ the outward pointing normal vector of $W_i$, maximizes the performance criterion (\ref{eq:criterion}) monotonically along the nodes' trajectories, leading in a locally optimal configuration.
\end{theorem}

\begin{proof}
We are interested in a gradient--based control law in the form
\begin{equation*}
	u_{i,q} = \alpha_{i,q} \frac{\partial\mathcal{H}}{\partial q_i}, ~~ u_{i,z} = \alpha_{i,z} \frac{\partial\mathcal{H}}{\partial z_i}, ~~ \omega_i = \alpha_{i,\theta} \frac{\partial\mathcal{H}}{\partial \theta_i}
\end{equation*}
in order to guarantee monotonic increase of the optimization criterion $\mathcal{H}$.

We will begin by evaluating $\frac{\partial\mathcal{H}}{\partial q_i}$. Through usage of the Leibniz integral rule, the neighbor definition in \ref{eq:neighbors} and since $\frac{\partial f(z_i)}{\partial q_i} = \frac{\partial f(z_j)}{\partial q_i} = 0$, we obtain
\begin{small}
\begin{equation*}
	\frac{\partial\mathcal{H}}{\partial q_i}
	= ~\int\limits_{\partial W_i} \upsilon_i^i ~n_i ~f(z_i) ~\phi(q) ~dq ~+ \sum\limits_{j \in N_i}
	~\int\limits_{\partial W_i \cap \partial W_j} \upsilon_j^i ~n_j ~f(z_j) ~\phi(q) ~dq
\end{equation*}
\end{small}
where $\upsilon_j^i$ stands for the Jacobian matrix with respect to $q_i$ of the points $q \in \partial W_j$,
\begin{equation}
	\upsilon_j^i\left(q\right) \stackrel{\triangle}{=} \frac{\partial q}{\partial q_i},~~q\in \partial W_j,~i,j\in I_n
	\label{eq:jacobian_xy}
\end{equation}
and $n_i$ for the outwards unit normal vector on $\partial W_i$.

The boundary $\partial W_i$ can be decomposed into disjoint sets as
\begin{align}
	\nonumber
	\partial W_i &=
	\{\partial W_i \cap \partial \Omega \}
	\cup
	\{\partial W_i \cap \partial \mathcal{O} \}
	\cup \\
	&\{\bigcup_{j \neq i} \left( \partial W_i \cap \partial W_j \right)\}
	\cup
	\{\overset{L}{\underset{l = 1}{\bigcup}} \left( \partial W_i \cap \partial W_c^l \right)\}.
	\label{eq:boundary_decomposition}
\end{align}
These sets represent the parts of $\partial W_i$ that lie on the boundary of $\Omega$, the boundary of the node's sensing region, the parts that are common between the boundary of the cell of node $i$ and those of other nodes and finally the parts that lie on the boundary of some unassigned region covered by multiple nodes. This decomposition can be seen in Figure~\ref{fig:boundary_decomposition} with the sets $\partial W_i \cap \partial \Omega$, $\partial W_i \cap \partial \mathcal{O}$, $\bigcup_{j \neq i} \left( \partial W_i \cap \partial W_j \right)$ and $\overset{L}{\underset{l = 1}{\bigcup}} \left( \partial W_i \cap \partial W_c^l \right)$ appearing in solid purple, red, green and blue respectively.

\begin{figure}[htb]
	\centering
	\includegraphics[width=0.3\textwidth]{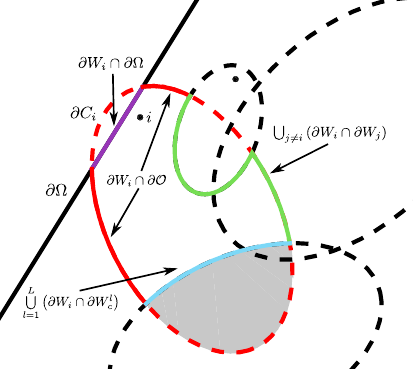}
	\caption{Decomposition of $\partial W_i$ into disjoint sets.}
	\label{fig:boundary_decomposition}
\end{figure}

At $q \in \partial \Omega$ it holds that $\upsilon_i^i = \textbf{0}_{2 \times 2}$ since we assume the region of interest is static. In addition, at $q \in \partial W_i \cap \partial W_c^l$ it holds that $\upsilon_i^i = \textbf{0}_{2 \times 2}$ since $\partial W_i \cap \partial W_c^l$ are arcs of some sensed region $\partial C_j$ and thus independent of the state of node $i$. Combining this with the boundary decomposition (\ref{eq:boundary_decomposition}) we obtain
\begin{align*}
	\frac{\partial\mathcal{H}}{\partial q_i}
	&= ~\int\limits_{\partial W_i \cap \partial \mathcal{O}} \upsilon_i^i ~n_i ~f(z_i) ~\phi(q) ~dq \\
	&+ \sum\limits_{j \in N_i} ~\int\limits_{\partial W_i \cap \partial W_j} \upsilon_i^i ~n_i ~f(z_i) ~\phi(q) ~dq \\
	&+ \sum\limits_{j \in N_i} ~\int\limits_{\partial W_j \cap \partial W_i} \upsilon_j^i ~n_j ~f(z_j) ~\phi(q) ~dq.
\end{align*}
Because the boundary $\partial W_i \cap \partial W_j$ is common among nodes $i$ and $j$, it holds that $\upsilon_j^i = \upsilon_i^i$ when evaluated over it and that  $n_j = -n_i$. Finally, the sums and the integrals within them can be combined and the final form of the planar control law is produced
\begin{align*}
	\frac{\partial\mathcal{H}}{\partial q_i}
	&= ~\int\limits_{\partial W_i \cap \partial \mathcal{O}} \upsilon_i^i ~n_i ~f(z_i) ~\phi(q) ~dq ~+\\
	&\qquad \sum\limits_{j \in N_i} ~\int\limits_{\partial W_i \cap \partial W_j} \upsilon_i^i ~n_i ~\left(f(z_i) - f(z_j)\right) ~\phi(q) ~dq.
\end{align*}

By following a similar procedure, defining
\begin{equation}
	\nu_j^i\left(q\right) \stackrel{\triangle}{=} \frac{\partial q}{\partial z_i},~~q\in \partial W_j,~i,j\in I_n
	\label{eq:jacobian_z}
\end{equation}
and
\begin{equation}
	\tau_j^i\left(q\right) \stackrel{\triangle}{=} \frac{\partial q}{\partial \theta_i},~~q\in \partial W_j,~i,j\in I_n,
	\label{eq:jacobian_theta}
\end{equation}
using the boundary decomposition (\ref{eq:boundary_decomposition}) and the fact that $\frac{\partial f(z_j)}{\partial z_i} = \frac{\partial f(z_i)}{\partial \theta_i} = \frac{\partial f(z_j)}{\partial \theta_i} = 0$ we obtain
\begin{small}
\begin{align*}
	\frac{\partial\mathcal{H}}{\partial z_i} &= \int\limits_{W_i} f^d(z_i) ~\phi(q) ~dq ~+ \int\limits_{\partial W_i \cap \partial \mathcal{O}} \nu_i^i \cdot n_i ~f(z_i) ~\phi(q) ~dq ~+ \\
	&\quad \sum\limits_{j \in N_i} \left[ ~\int\limits_{\partial W_i \cap \partial W_j} \nu_i^i \cdot n_i ~\left(f(z_i)-f(z_j)\right) ~\phi(q) ~dq \right]
\end{align*}
\end{small}
and
\begin{small}
\begin{align*}
	\frac{\partial\mathcal{H}}{\partial \theta_i} &= \int\limits_{\partial W_i \cap \partial \mathcal{O}} \tau_i^i \cdot n_i ~f(z_i) ~\phi(q) ~dq ~+ \\
	&\quad \sum\limits_{j \in N_i} \left[ ~\int\limits_{\partial W_i \cap \partial W_j} \tau_i^i \cdot n_i ~\left(f(z_i)-f(z_j)\right) ~\phi(q) ~dq \right]
\end{align*}
\end{small}

\end{proof}

\section{Simulation Studies}
\label{sec:simulations}
Simulation studies used to evaluate the efficiency of the proposed control
scheme are presented in this section. The region of interest $\Omega$ is chosen
to be the same as in \cite{Stergiopoulos_IETCTA10} for consistency. All nodes
are identical with an elliptical base sensing pattern and altitude constraints
$z^{\min} = 0.3$ and $z^{\max} = 2.3$. The nodes' initial states, excluding
their altitude, are the same as in \cite{Stergiopoulos_ICRA14}. The boundaries
of the nodes' cells are shown in solid black and the boundaries of their
sensing regions in dashed red lines. It should be noted that although for some
sensing patterns, such as circular ones, it is possible to derive the stable
altitude for a single MAA, this is not the case with elliptical sensing
patterns. This is due to the fact that there exists no closed--form expression
for the arc length of an ellipse. The simulations used in this section have
been implemented in Matlab and are available as free and open--source software
in
\url{https://git.sr.ht/~sotirisp/uav-coverage}.
 The code is written from the scope of a single agent, thus it can be easily
modified and used as a basis for a real world implementation.

\subsection{Case Study I}
In this simulation study the elliptical sensing patterns are approximated by their largest inscribed disks centered at $q_i$ with radii $r_b^{\min}$. In that case, only control laws (\ref{eq:control_xy}) and (\ref{eq:control_z}) are used, since the sensing patterns are rotation invariant. The resulting control scheme is the one examined in \cite{Papatheodorou2017_RAS}. Figure \ref{fig:states_traj} (top) shows the evolution of the network through time, with the real elliptical sensing patterns shown by dotted black lines. The coverage quality at the initial and final configurations of the network are shown in Figure \ref{fig:sim1_quality} while the value of the optimization criterion and the percentage of covered area are shown in Figure \ref{fig:objective_area} in blue.

\begin{figure}[htbp]
	\centering
	\includegraphics[width=0.23\textwidth]{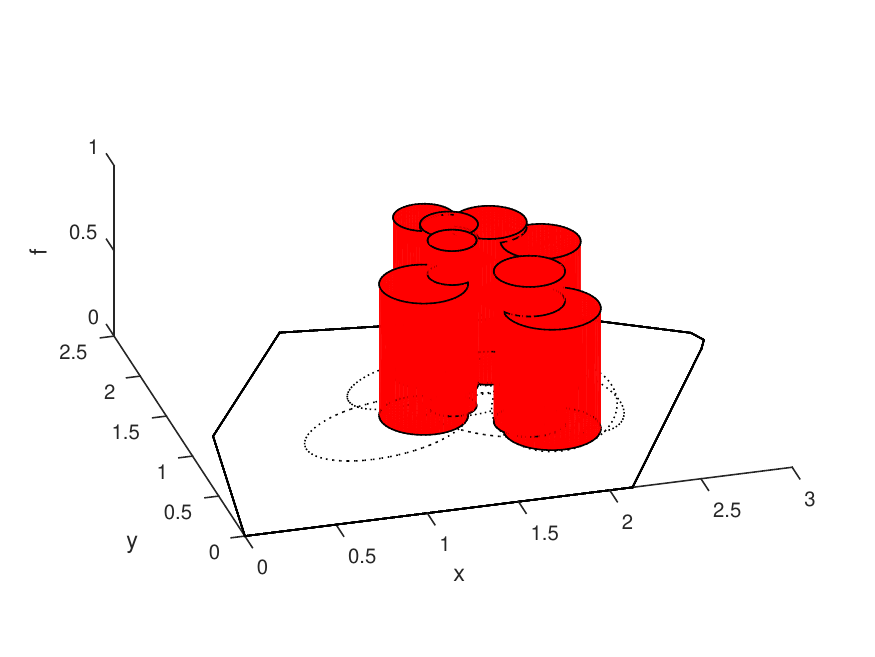}
	\includegraphics[width=0.23\textwidth]{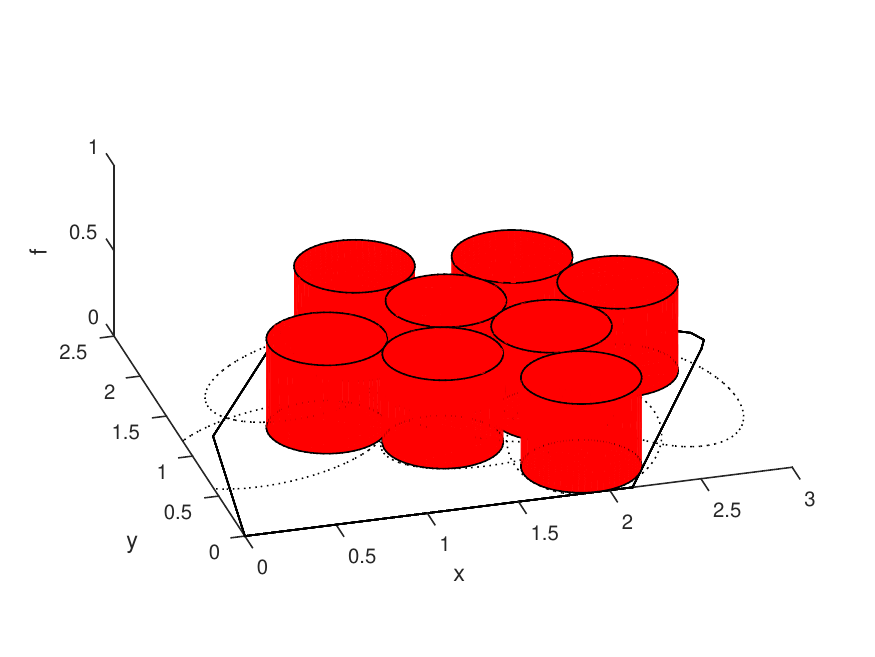}
	\caption{Case Study I: Initial [Left] and final coverage quality.}
	\label{fig:sim1_quality}
\end{figure}

\subsection{Case Study II}
In this simulation study the exact elliptical sensing patterns of the nodes are used, however their orientations remain fixed to their initial values. Thus, again only control laws (\ref{eq:control_xy}) and (\ref{eq:control_z}) are used. Figure \ref{fig:states_traj} (middle) shows the evolution of the network through time while the coverage quality at the initial and final configurations of the network are shown in Figure \ref{fig:sim2_quality}. The value of the optimization criterion and the percentage of covered area are shown in Figure \ref{fig:objective_area} in red.

\begin{figure}[htbp]
	\centering
	\includegraphics[width=0.23\textwidth]{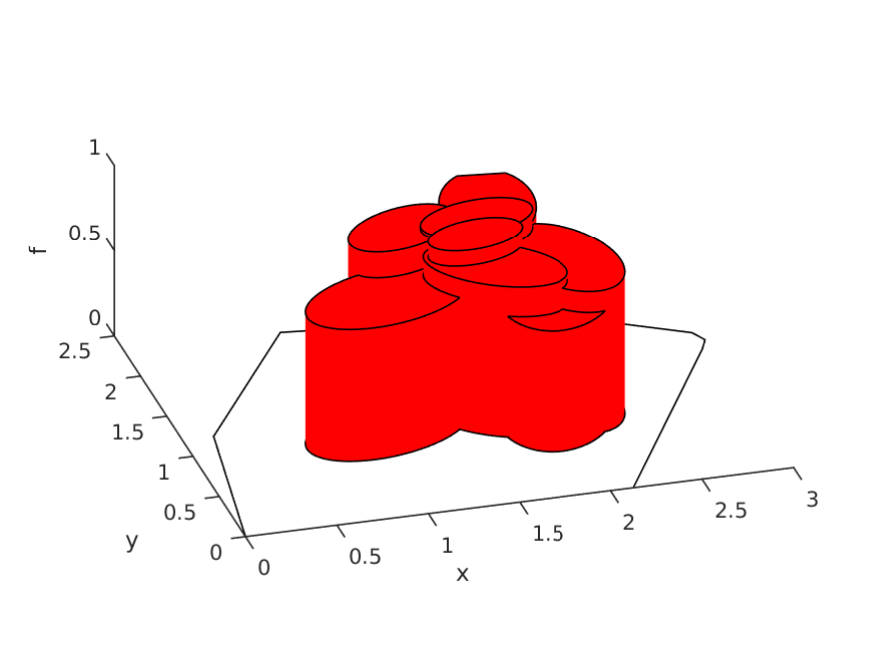}
	\includegraphics[width=0.23\textwidth]{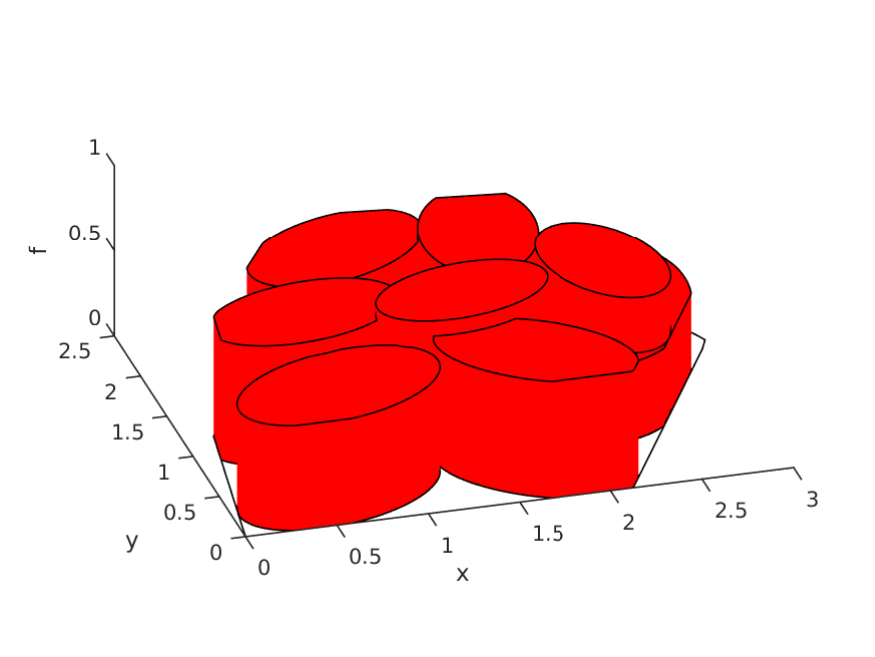}
	\caption{Case Study II: Initial [Left] and final coverage quality.}
	\label{fig:sim2_quality}
\end{figure}

\subsection{Case Study III}
In this final simulation study the complete control scheme (\ref{eq:control_xy}), (\ref{eq:control_z}) and (\ref{eq:control_theta}) is used in conjunction with the elliptical sensing patterns. Figure \ref{fig:states_traj} (middle) shows the evolution of the network through time while the coverage quality at the initial and final configurations of the network are shown in Figure \ref{fig:sim3_quality}. The value of the optimization criterion and the percentage of covered area are shown in Figure \ref{fig:objective_area} in green.

It is clear from Figure \ref{fig:objective_area} that a maximal inscribed circle approximation of the sensed region can severely limit the efficiency of the agents. Even when the agents are not allowed to rotate, the network's performance is significantly improved when the exact sensing patterns are used compared to their circular approximations. However, in such a case the agents are not always used to their full potential, as can be seen in Figure \ref{fig:states_traj} (middle) [Right] where the sensing region of an agent is completely contained within that of another.

\begin{figure}[htbp]
	\centering
	\includegraphics[width=0.23\textwidth]{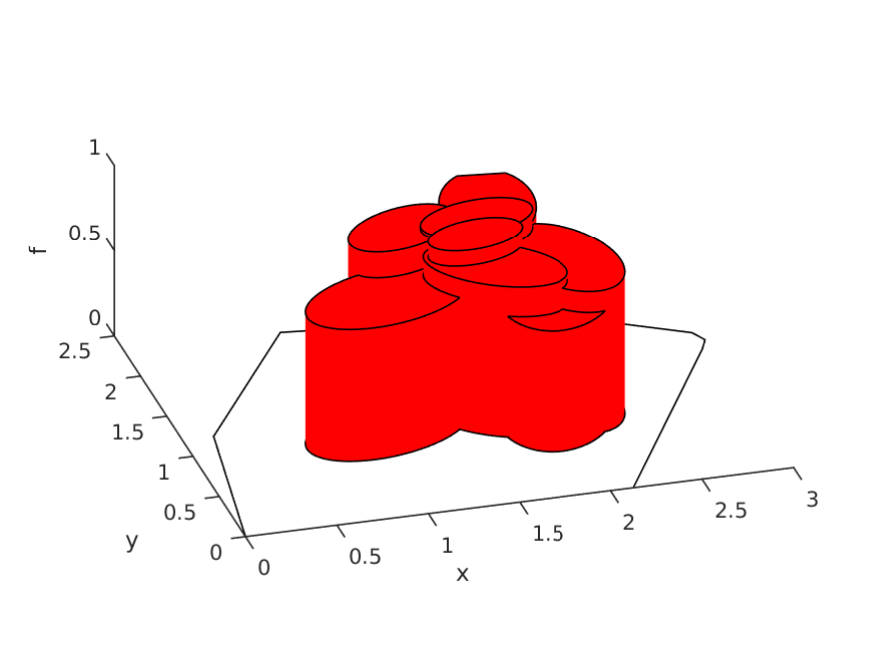}
	\includegraphics[width=0.23\textwidth]{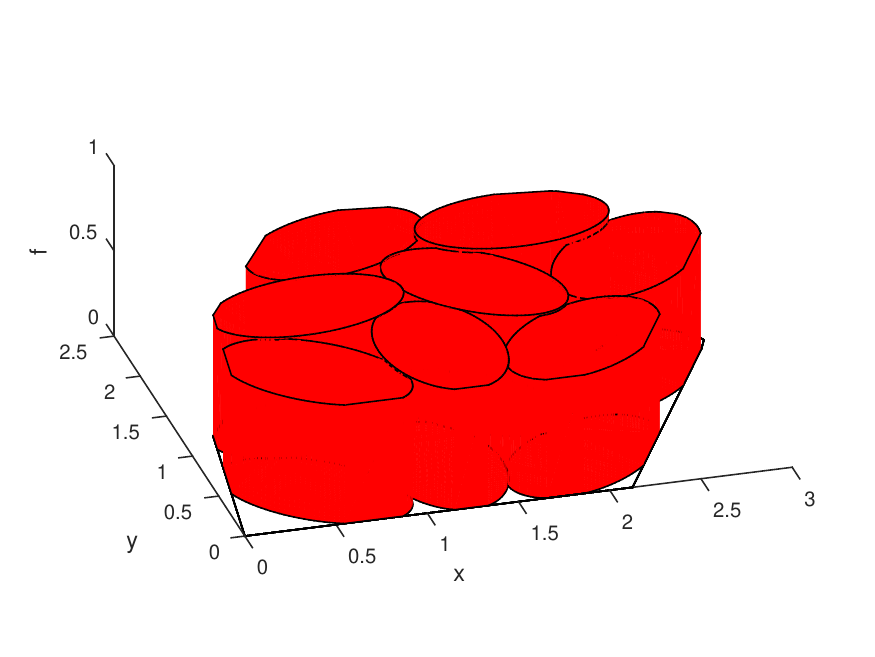}
	\caption{Case Study III: Initial [Left] and final coverage quality.}
	\label{fig:sim3_quality}
\end{figure}

\begin{figure*}[htbp]
	\centering
	\includegraphics[width=0.31\textwidth]{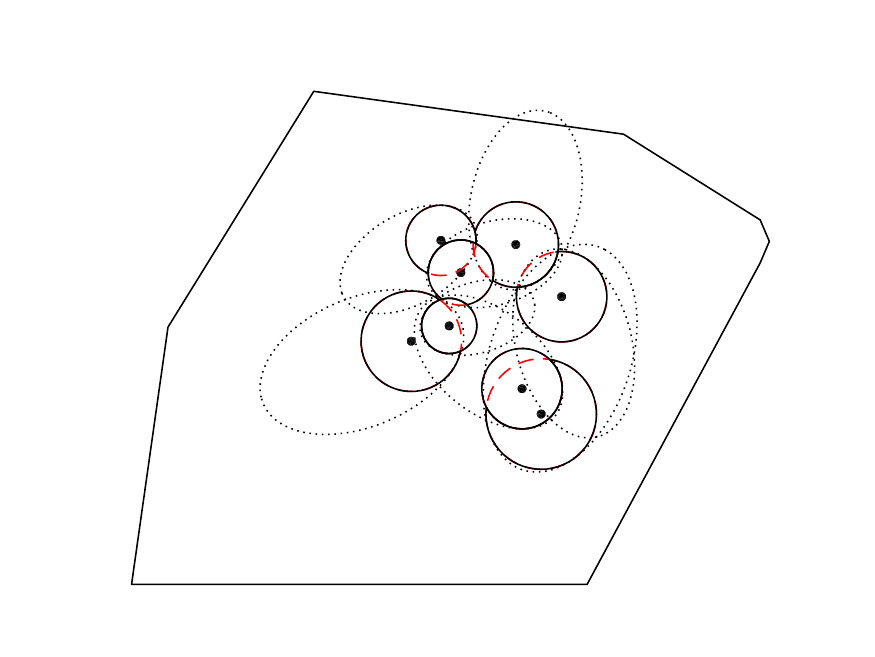}
	\includegraphics[width=0.31\textwidth]{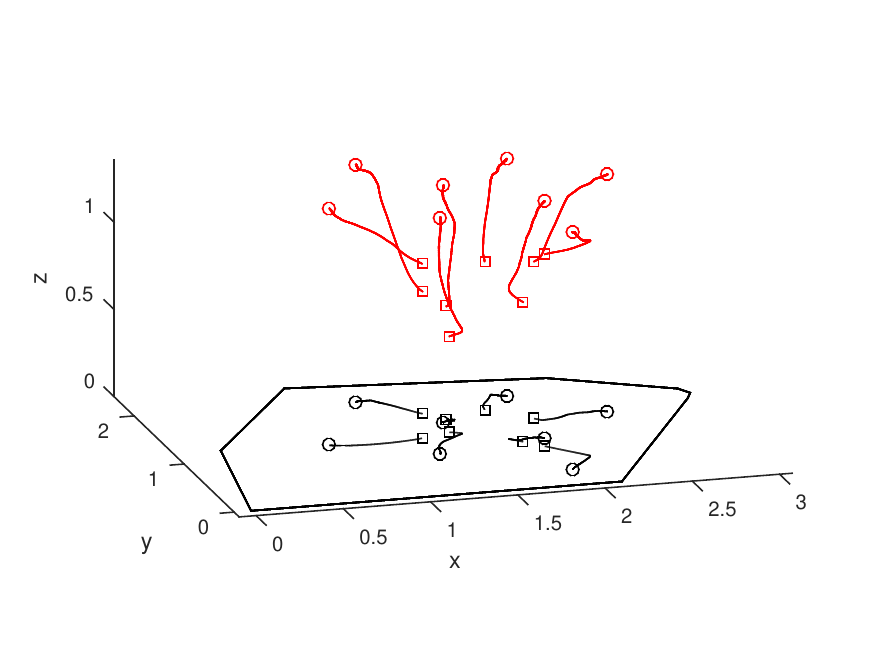}
	\includegraphics[width=0.31\textwidth]{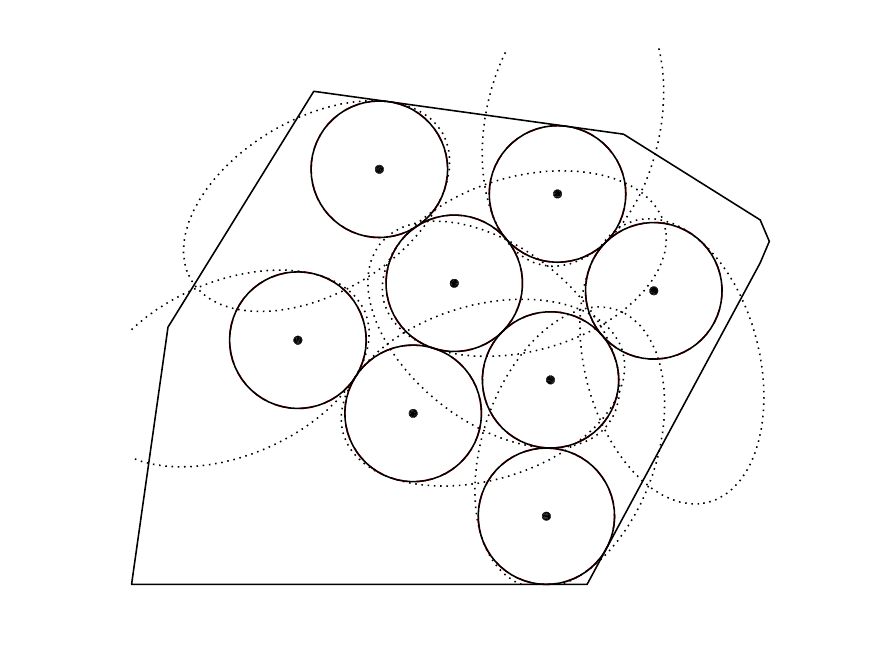} \\
	\includegraphics[width=0.31\textwidth]{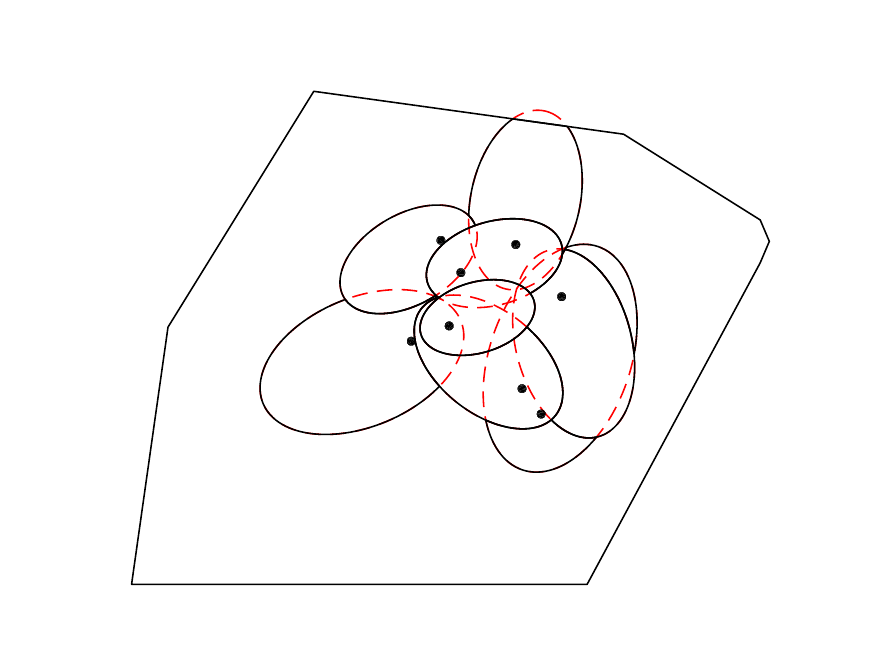}
	\includegraphics[width=0.31\textwidth]{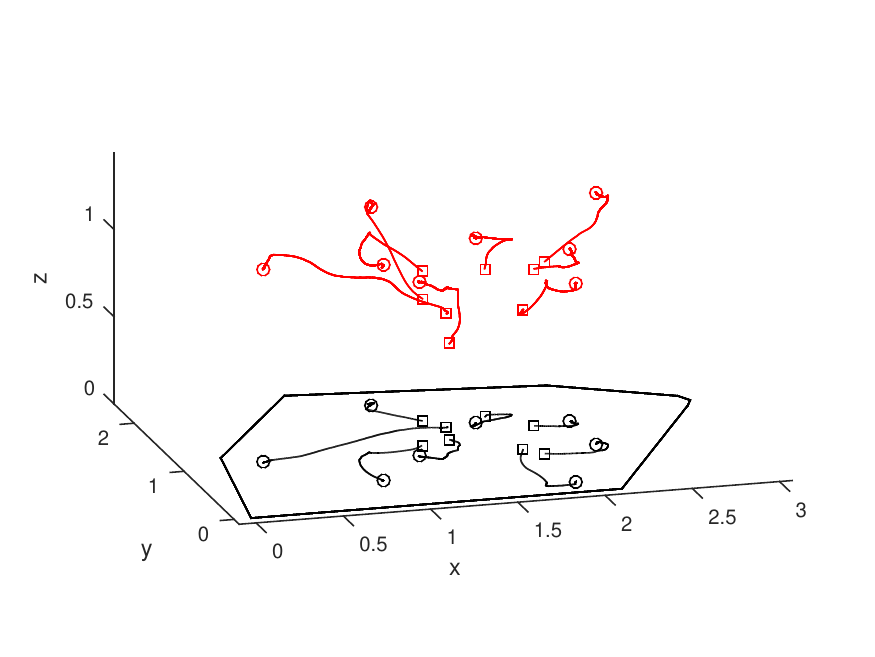}
	\includegraphics[width=0.31\textwidth]{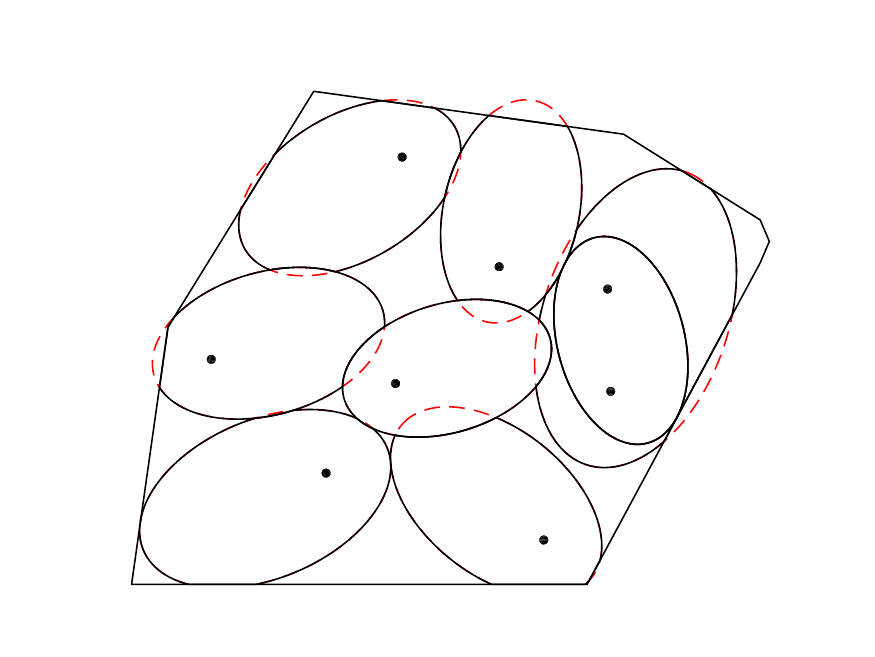} \\
	\includegraphics[width=0.31\textwidth]{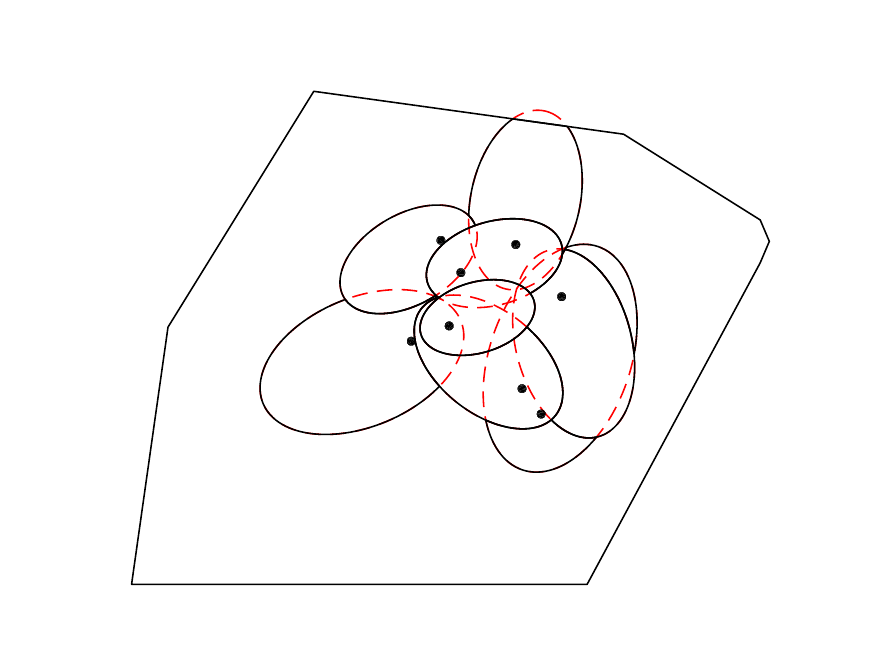}
	\includegraphics[width=0.31\textwidth]{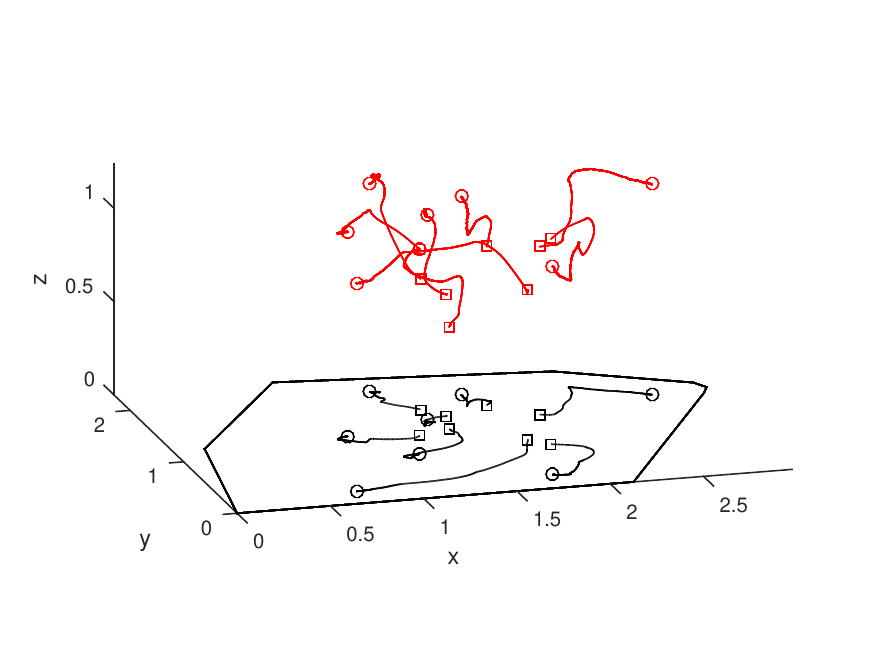}
	\includegraphics[width=0.31\textwidth]{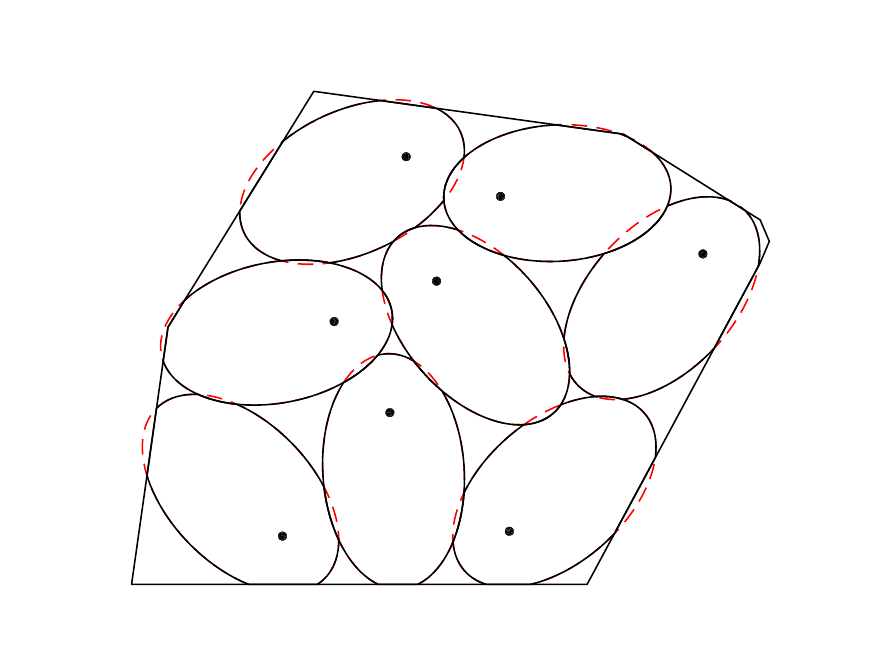}
	\caption{Initial network state [Left], agent trajectories [Center] and final network state [Right] for Case Study I (top), Case Study II (middle) and Case Study III (bottom).}
	\label{fig:states_traj}
\end{figure*}

\begin{figure}[htbp]
	\centering
	\includegraphics[width=0.23\textwidth]{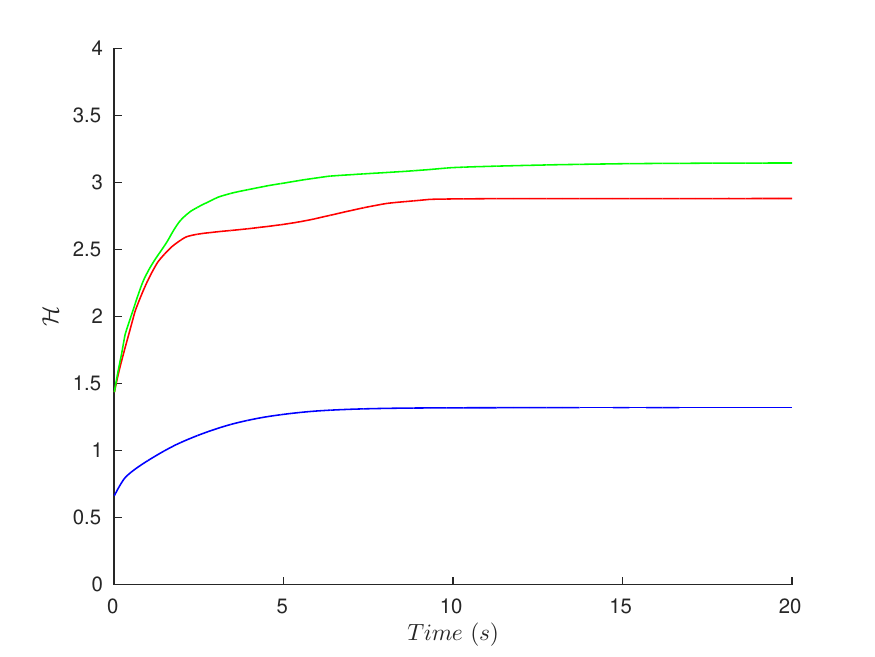}
	\includegraphics[width=0.23\textwidth]{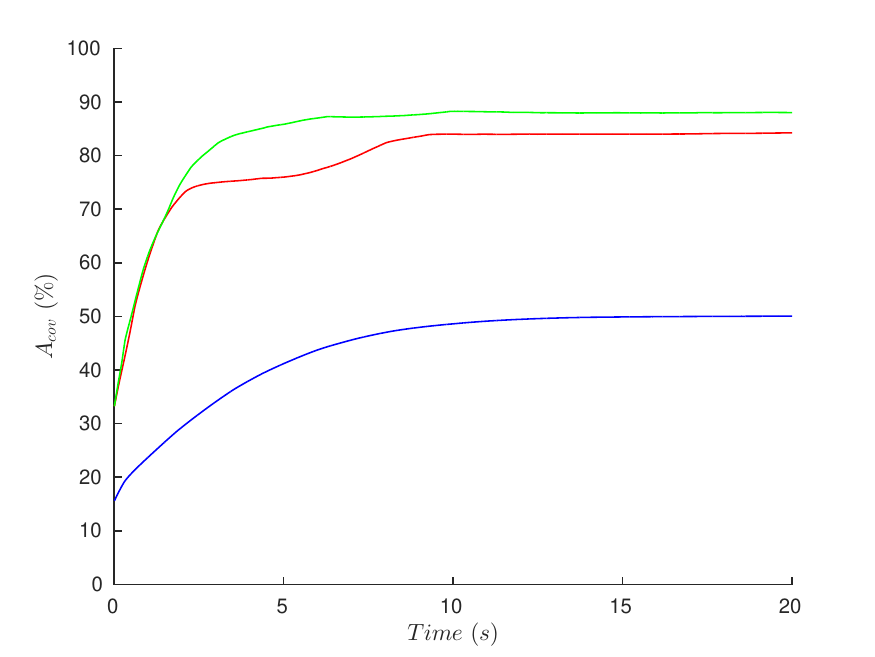}
	\caption{Coverage--quality objective $\mathcal{H}$ [Left] and percentage of $\Omega$ covered [Right].}
	\label{fig:objective_area}
\end{figure}

\section{Conclusions}
Area coverage by a network of MAAs has been examined in this article. Each MAA is equipped with a downwards facing camera with a generic convex sensing pattern. A partitioning of the region sensed by all MAAs is used to assign regions of responsibility to each one of them and based on them a gradient--based control law is derived. The control law maximizes a combined coverage--quality optimization criterion. Simulation studies were provided in order to evaluate the proposed control scheme and compare its performance with that of other simplified control strategies for the same problem.

\bibliographystyle{IEEEtran}
\bibliography{references}
\end{document}